\documentclass[12pt,a4paper,normalheadings,headsepline,headinclude,bibtotoc,fleqn,review]{article}
\usepackage[latin1]{inputenc}
\usepackage{amsmath}
\usepackage{booktabs}
\usepackage{array}
\usepackage{amsthm}
\usepackage{dcolumn}
\usepackage{endnotes}
\usepackage{units}
\usepackage{marvosym}
\usepackage[lines=10]{geometry}
\usepackage{longtable}
\usepackage{rotating}
\usepackage{expdlist}
\usepackage{wrapfig}
\usepackage{expdlist}
\usepackage{amssymb}
\usepackage{gloss}
\usepackage{nomencl}
\usepackage{natbib}

\newtheorem{prop}{Proposition}
\newtheorem{coro}{Corollary}

    \makegloss

\renewcommand{\thesection}{\arabic{section}}

\usepackage{fancyhdr} 
\begin{document}

\begin{center}\huge{A Global Game with Heterogenous Priors}\footnote{I thank Martin Hellwig for bringing the literature on global games to my attention.
I thank Sophie Bade, Nataliya Demchenko, Alia Gizatulina, Olga
Gorelkina, Dominik Grafenhofer, Carl Christian von Weizs�cker, and
seminar participants in Bonn for helpful and encouraging
discussions.}\end{center}

\begin{center} \emph{Wolfgang Kuhle}\\\emph{Max Planck Institute for Research on Collective Goods, Kurt-Schumacher-Str. 10, 53113 Bonn, Germany. Email:
kuhle@coll.mpg.de.}
\end{center}

\noindent\emph{\textbf{Abstract:} This paper relaxes the common
prior assumption in the public and private information game of
\citet{Mor01,Mor04}. For the generalized game, where the agent's
prior expectations are heterogenous, it derives a sharp condition
for the emergence of unique/multiple equilibria. This condition
indicates that unique equilibria are played if player's public
disagreement is substantial. If disagreement is small, equilibrium
multiplicity depends on the relative precisions of private signals
and subjective priors. Extensions to environments with public
signals of exogenous and endogenous quality show that prior
heterogeneity, unlike heterogeneity in private information,
provides a robust anchor for unique equilibria. Finally,
irrespective of whether priors are common or not, we show
that public signals can ensure equilibrium uniqueness, rather than multiplicity, if they are sufficiently precise.}

\section{Introduction}\label{sg1}

Heterogeneity in private beliefs is the core component of global
coordination games. In the original two-player games introduced by
\citet{Car93} and \citet{Rub89}, as well as in the \citet{Mor98}
extension with a continuum of players, it is a small perturbation
away from common knowledge, which selects unique equilibria. This
pivotal role of private beliefs was put into perspective by
\citet{Mor01,Mor04}, \citet{Hel02}, and \citet{Met02}, who
introduce common priors and public signals into the global games
model. In such extended settings, it turns out that the global
game structure, and its inductive equilibrium
selection mechanism, is distorted. 
And it depends on the relative precision of public and private
information whether agents can coordinate on multiple equilibria.

Given that the dispersion of private beliefs is pivotal to the
global games approach, it is complementary to study how a
relaxation of the common prior assumption, which adds a new
dimension to the dispersion of private beliefs, affects
equilibrium outcomes. To address this question, we study
coordination games which are played by a continuum of agents that
differ with regard to (i) private information and (ii) the mean of
their prior over an unknown fundamental that characterizes the
game's payoffs. While the heterogeneity in private signals is
unobservable to players, players will know the distribution of
priors. That is, we study how agents can coordinate on equilibria
in environments, where they know that others hold, and act upon,
views that they commonly know to disagree with. The main result
shows that if player's public disagreement regarding the
coordination game's payoff structure is substantial, then they
play unique equilibria. That is, unique equilibria can be selected
through the dispersion of priors.

We derive our results in two steps. First, we isolate the pure
interaction between private signals and dispersed priors in an
environment which abstracts from public signals. Second, we add
three different types of public signals to examine the robustness
of our results. In the first environment, which abstracts from
public signals, we find that \emph{heterogeneity in beliefs, which
originates from the dispersion of prior expectations $\mu$,
contributes to equilibrium uniqueness. Dispersion in private
signals, on the contrary,
induces equilibrium multiplicity.} 
In the extended settings, where we introduce different types of
public signals, we find that prior heterogeneity still
unambiguously contributes towards equilibrium uniqueness. The role
of private signals, on the contrary, is ambiguous and depends
critically on the specification of the public signal. The main
finding in the second part of the paper is therefore that
\emph{heterogeneity in priors, unlike heterogeneity in private
information, provides a robust anchor for unique equilibria}. Put
differently, we show that heterogeneous priors robustly select
unique equilibria in rather diverse environments. In turn, we
compare this finding to the literature where the presence of
common priors induces multiple rather than unique equilibria.

Finally, as a byproduct of our analysis, we find that public
signals in themselves, irrespective of whether priors are
heterogenous or common, have an ambiguous effect on equilibrium
multiplicity: increases in the public signal's precision can
ensure equilibrium uniqueness. This result is of interest in
comparison with those of \citet{Mor01,Mor04}, \citet{Hel02},
\citet{Met02}, \citet{Ang06} who find that increases in the public
signal's precision unambiguously induce multiple rather than
unique equilibria.\footnote{There will be two classes of
equilibria in this setting, and the public signal's precision
ensures uniqueness \emph{for every given signal realization}.
However, if the public signal's precision is sufficiently high,
there may be multiplicity in the public signal's realization
itself.}

Compared to the literature, we extend the canonic common prior
coordination games by \citet{Mor01,Mor04} to a setting where
agents' prior distribution's mean $\mu$, regarding the true state
of the world, are normally distributed over the economy's
population. In the model of \citet{Mor01,Mor04}, equilibrium
uniqueness is ensured iff
$\frac{\sqrt{\alpha_x}}{\alpha_p}\geq\frac{1}{\sqrt{2\pi}}$; where
$\alpha_x$ is the private signal's precision and $\alpha_p$ is the
prior's precision. For the generalized game, where subjective
prior expectations $\mu$ are distributed with variance
$\sigma_{\mu}$, we obtain
$\sqrt{(\frac{\sqrt{\alpha_x}}{\alpha_p})^2+\sigma^2_{\mu}}\geq\frac{1}{\sqrt{2\pi}}$
as an analog condition for equilibrium uniqueness. This formula's
main strength is that it isolates how prior heterogeneity
$\sigma_{\mu}$ affects the number of possible equilibria. In
general, the condition shows that incomplete information games
where the population of players is polarized such that there are
two large tails have unique equilibria. Or, equivalently, unique
equilibria are ensured if there is only a small group of agents
which have ``intermediate" or ``moderate" beliefs. Put
differently, increases in the priors dispersion, $\sigma_{\mu}$,
``thin out" the group of agents with ``moderate" beliefs which
can, potentially, coordinate on multiple equilibria. Technically,
equilibrium multiplicity depends on the relative precisions of
private information $(\alpha_x)$ and the subjective prior
$(\alpha_p)$ if the prior's dispersion is small
$(\sigma_{\mu}<\frac{1}{\sqrt{2\pi}})$. For sufficiently dispersed
priors, $(\sigma_{\mu}>\frac{1}{\sqrt{2\pi}})$, equilibria are
unique, irrespective of the relative precisions of private signals
and priors. The original uniqueness condition for the common prior
economy
$(\frac{\sqrt{\alpha_x}}{\alpha_p}\geq\frac{1}{\sqrt{2\pi}})$
obtains as prior dispersion $(\sigma_{\mu})$ vanishes.

In Section \ref{s45}, we extend the information structure and
embed our coordination game into three different public signal
environments. These extensions provide a background to study the
fundamental difference between heterogeneity in priors and private
signals. The main finding in this section is that \emph{only
heterogeneity in priors selects unique equilibria reliably}. On
the contrary, the role of both private and public signals changes
from environment to environment. These findings originate, first
of all, from an environment where the public signal reveals the
true state of the game with exogenous quality; secondly, from an
environment that may be seen as a financial markets context, where
stock prices aggregate and reveal dispersed private information
with endogenous quality; and thirdly, from an environment with
public signals that allow agents to observe each other's actions.
Comparison of all three public signal settings shows that the
prior's dispersion $\sigma_{\mu}$ is the only parameter which can
ensure unique equilibria in all environments. The implications of
the public and private signals' precisions, however, vary from
case to case. In particular, increases in the private signal's
precision ensure uniqueness (multiplicity) when public signals are
of exogenous (endogenous) precision.

Recently, \citet{Ste08}, \citet{Izm10}, and \citet{Mat12} have
introduced heterogenous priors into global games. \citet{Ste08},
\citet{Izm10}, and \citet{Mat12} focus, respectively, on learning,
rationalizability of strategies, and applications to mechanism
design. Unlike \citet{Ste08} and \citet{Mat12}, who focus on
N-player games, we study games with a continuum of players as in
\citet{Izm10}. Contrary to \citet{Izm10}, who focus on the private
information limit\footnote{In the ``private information limit",
the private signal's precision goes to infinity. A crucial
consequence of this is that the importance of priors for the
agent's decisions vanishes.} of \citet{Car93}, which carries over
to these more general games, the present paper develops a sharp
characterization of equilibrium multiplicity at and away from the
private information limit.\footnote{\citet{Izm10} discuss prior
heterogeneity in a game with a continuum of players.
\citet{Izm10}, p. 25, focus on the formation of individual
threshold strategies and avoid the ``delicate" question of
equilibrium multiplicity, which is the focus of the present paper.
More precisely, \citet{Izm10} study games with heterogenous priors
(``sentiments") and characterize for a family of coordination
games how prior dispersion affects the unique threshold
equilibrium which is ensured once private information becomes
sufficiently precise. Due to the focus on the private information
limit, it remains an open question how prior heterogeneity affects
the number of equilibria in those cases where private signals are
not arbitrarily precise. This question is the focus of the current
paper. Finally, while \citet{Izm10} characterize games with
general distribution functions, the present characterization of
equilibrium uniqueness and multiplicity, away from the limit where
private information is very precise, relies on the afore mentioned
normality assumptions.} The present paper is therefore the first
to derive closed form expressions for the emergence of multiple
equilibria in large games with heterogeneous priors. Moreover, by
means of the public signal environments, it provides a framework
to interpret the role of prior heterogeneity in the global games
structure.\footnote{\citet{Ste08}, \citet{Izm10}, and
\citet{Mat12} abstract from public signals.}

Regarding the heterogeneity in priors, we note that
\citet{Mor98,Mor01,Mor04} embed their common prior incomplete
information coordination games in a currency crises, bank-run and
debt-run context respectively. In these contexts, our assumption
of a prior distribution
$\mu_i\sim\mathcal{N}(E[\mu],\sigma^2_{\mu})$, which is known to
all agents, may be interpreted as a publicly observable
distribution of exchange rate forecasts or credit ratings. Such a
distribution of conflicting beliefs, may also be seen as
consistent with publicly disclosed long and short positions that
large investors take in a firm's stock or debt. While
\citet{Mor95} and \citet{Set12}, discuss heterogenous priors and
their origins in detail, we will take prior heterogeneity as a
given. In turn, we focus purely on the implications that public
disagreement has for the emergence of multiple equilibria in the
global games framework.


The rest of the paper is organized as follows. In Section
\ref{s2}, we introduce the model. Section \ref{s3} contains the
main result. In Section \ref{s45}, we reflect on our findings in
three distinct public signal environments. Section \ref{s5}
concludes.



\section{The Model}\label{s1}

To isolate the implications of prior dispersion, we start with a
setting that differs only with regard to the heterogeneity of
priors from the canonic framework introduced by
\citet{Mor01,Mor04}. 

\subsection{Strategies, Payoffs and Information}\label{s2} There is a
status quo and a unit measure of agents indexed by $i\in[0,1]$.
Each of these agents $i$ can choose between two actions
$a_i\in\{0,1\}$. Choosing $a_i=1$ means to attack the status quo.
Choosing $a_i=0$ means that the agent does not attack the status
quo. An attack on the status quo is associated with a cost
$c\in(0,1)$. If the attack is successful, the status quo is
abandoned and attacking agents receive a net payoff $1-c>0$. If
the attack is not successful, an attacking agent's net payoff is
$-c$. The payoff for an agent who does not attack is normalized to
zero.

The status quo is abandoned if the aggregate size of the attack
$A:=\int_0^1a_idi$ exceeds the strength of the status quo
$\theta$, i.e., if $A>\theta$. Otherwise, if $A<\theta$, the
status quo is maintained and the attack fails. Regarding the
fundamental, $\theta$, each agent $i$ holds a prior belief
$\theta\sim\mathcal{N}(\mu_i,\sigma^2_p)$. In our setting, we
assume that the priors, regarding the fundamental $\theta$, are
normally distributed across the population, i.e.,
$\mu_i\sim\mathcal{N}(E[\mu],\sigma^2_{\mu})$. And the
distribution of priors is common knowledge. In addition to the
prior, each agent receives a private signal
$x_i=\theta+\sigma_{x}\xi_i$. Where signal noise
$\xi_i\sim\mathcal{N}(0,1)$ is i.i.d. across the population. We
parameterize the information structure in terms of precisions
$\alpha_x\equiv\frac{1}{\sigma_x^2}$ and
$\alpha_p\equiv\frac{1}{\sigma_p^2}$. Agents use their information
to calculate expected utility
\begin{eqnarray} E[U(A,\theta,c,a_i)|\mu_i,x_i]=a_i(1\times P(\theta <A|\mu_i,x_i)-c).     \label{1} \end{eqnarray}
Where $P(\theta<A|\mu_i,x_i)$ is the probability that the attack
is successful given signal $x_i$ and prior belief $\mu_i$.
Regarding payoffs, (\ref{1}) reflects that an agent $i$, who does
not attack $(a_i=0)$ receives a safe payoff of $0$, while an agent
who does attack $(a_i=1)$ receives $1-c$ if the attack is
successful, and $-c$ otherwise. To close the description of
agents' information, we note that the distributions of $x_i$ and
$\mu_i$
are commonly known to all players. 
With the private choice problem in place, we now characterize
equilibrium.


\subsection{Equilibrium}\label{s3}

As \citet{Mor98,Mor04}, we focus on monotone threshold equilibria.
Each threshold equilibrium will be characterized by a pair
$\psi^*,\theta^*$. Where $\theta^*$ separates values
$\theta<\theta^*$ for which the status quo is abandoned from
values $\theta>\theta^*$ where it prevails. The threshold level
$\psi^*\equiv\frac{\alpha_x}{\alpha}x^*+\frac{\alpha_p}{\alpha}\mu^*$
summarizes the critical pairs $x^*,\mu^*$, for which agents are
just indifferent between attacking and not attacking.\footnote{In
the context of the currency crises model of \citet{Mor98},
attacking agents would sell short a country's currency, and the
central bank's reserves $\theta$ are either sufficient
$(\theta>A)$ to defend the peg  or not $(\theta<A)$. In another
interpretation, agents can run on/sell short a firm's debt, and if
the firm's financial strength is insufficient it defaults.} The
equilibrium pairs, $\psi^*,\theta^*$, are determined by the
simultaneous evaluation of the \emph{payoff indifference
condition} and the \emph{critical mass condition}. The
\emph{payoff indifference condition}, PIC, follows directly from
the individual choice problem (\ref{1}). Taking the threshold
level $\theta^*$ as given, we have:\footnote{See \citet{Rai61}, p.
250, for the standard results on prior and posterior distributions
of normally distributed variables which are used throughout the
paper.}\footnote{Note that we already use the (forthcoming)
\emph{critical mass condition}, (\ref{cmc11}), which requires that
$\theta^*\equiv A(\psi^*,\theta^*)$, and replace $A$ with
$\theta^*$ in (\ref{1}) to obtain (\ref{pic11}).}
\begin{eqnarray} P(\theta\leq\theta^*|x^*,\mu^*)=c\quad\Leftrightarrow \quad
\Phi\Big(\sqrt{\alpha}(\theta^*-\frac{\alpha_x}{\alpha}x^*-\frac{\alpha_p}{\alpha}\mu^*)\Big)=c;
\quad \alpha=\alpha_x+\alpha_p, \label{pic11}\end{eqnarray} where
$\Phi()$ is the cumulative of the standard normal distribution.
The PIC (\ref{pic11}) characterizes those pairs $x^*,\mu^*$ which
are such that the agent is indifferent between attacking and not
attacking. To work with Condition (\ref{pic11}), we define
$\psi\equiv\frac{\alpha_x}{\alpha}x+\frac{\alpha_p}{\alpha}\mu$
such that (\ref{pic11}) writes:
\begin{eqnarray} P(\theta\leq\theta^*|\psi^*)=c\quad\Leftrightarrow \quad
\Phi\Big(\sqrt{\alpha}(\theta^*-\psi^*)\Big)=c; \quad
\alpha=\alpha_x+\alpha_p. \label{pic12}\end{eqnarray} And agents
who receive $\psi\leq\psi^*$, which is evidence of a weak
fundamental, attack. Agents who receive $\psi>\psi^*$ do not
attack since they believe in a strong fundamental, which makes a
successful attack unlikely.

The \emph{critical mass condition}, CMC, takes the cutoff value
$\psi^*$ as given and determines the threshold $\theta^*$, where
the attack is just strong enough to overwhelm the status quo. To
calculate the mass of attacking agents, we note that
$\psi|\theta\sim\mathcal{N}(\frac{\alpha_x}{\alpha}\theta+\frac{\alpha_p}{\alpha}E[\mu],\alpha_{\psi}^{-1})$
where
$\alpha_{\psi}\equiv\frac{\alpha^2}{\alpha_x+\alpha^2_p\sigma^2_{\mu}}$
and $\alpha=\alpha_x+\alpha_p$. The CMC therefore writes:
\begin{eqnarray} A(\psi^*,\theta^*)\equiv P(\psi<\psi^*|\theta^*)=\theta^*\quad
\Leftrightarrow\quad
\Phi\Big(\sqrt{\alpha_{\psi}}(\psi^*-\frac{\alpha_x}{\alpha}\theta^*-\frac{\alpha_p}{\alpha}E[\mu])\Big)=\theta^*.\label{cmc11}\end{eqnarray}
Again, $\Phi()$ is the cumulative of the standard normal
distribution. Regarding (\ref{cmc11}), we note that it implies
that there exists only one $\psi^*$ for every $\theta^*$.
\emph{Simultaneous} evaluation of (\ref{pic12}) and (\ref{cmc11})
yields threshold equilibria $\psi^*,\theta^*$:

\begin{prop}\label{p1}The equilibrium $\psi^*,\theta^*$ is unique, for all parameter pairs, if and only
if\\
$\sqrt{(\frac{\sqrt{\alpha_x}}{\alpha_p})^2+\sigma^2_{\mu}}\geq\frac{1}{\sqrt{2\pi}}$.\end{prop}

\begin{proof}We solve (\ref{pic12}) for the threshold level
$\psi^*=-\Phi^{-1}(c)\frac{1}{\sqrt{\alpha}}+\theta^*$, and
substitute $\psi^*$ into (\ref{cmc11}) to obtain a one-dimensional
equation in $\theta^*$:
\begin{eqnarray} \Phi\Big(\sqrt{\alpha_{\psi}}\Big(-\Phi^{-1}(c)\frac{1}{\sqrt{\alpha}}+\theta^*
-\frac{\alpha_x}{\alpha}\theta^*-\frac{\alpha_p}{\alpha}E[\mu]\Big)\Big)=\theta^*.\label{comp}\end{eqnarray}
The sufficient condition for uniqueness of the threshold
$\theta^*$ is therefore:
\begin{eqnarray}\sqrt{\alpha_{\psi}}\frac{\alpha_p}{\alpha}\phi\Big(\sqrt{\alpha_{\psi}}\Big(-\Phi^{-1}(c)\frac{1}{\sqrt{\alpha}}
+\frac{\alpha_p}{\alpha}\theta^*-\frac{\alpha_p}{\alpha}E[\mu]\Big)\Big)\leqq1\nonumber\\
\Leftrightarrow
\sqrt{\alpha_{\psi}}\frac{\alpha_p}{\alpha}\frac{1}{\sqrt{2\pi}}e^{-\Big(\sqrt{\alpha_{\psi}}\Big(-\Phi^{-1}(c)\frac{1}{\sqrt{\alpha}}
+\frac{\alpha_p}{\alpha}\theta^*
-\frac{\alpha_p}{\alpha}E[\mu]\Big)\Big)^2\frac{1}{2}}\leqq1.
\label{2}\end{eqnarray} Finally, we recall that
$\alpha_{\psi}=\frac{\alpha^2}{\sqrt{\alpha_x+\alpha^2_p\sigma^2_{\mu}}}$
and take logarithms to obtain: \begin{eqnarray}
ln\Big(\frac{1}{\sqrt{2\pi}}
\frac{1}{\sqrt{(\frac{\sqrt{\alpha_x}}{\alpha_p})^2+\sigma^2_{\mu}}}\Big)\leqq\Big(\sqrt{\alpha_{\psi}}\Big(-\Phi^{-1}(c)\frac{1}{\sqrt{\alpha}}
+\frac{\alpha_p}{\alpha}\theta^*
-\frac{\alpha_p}{\alpha}E[\mu]\Big)\Big)^2\frac{1}{2}.\nonumber
\end{eqnarray}Accordingly, unique equilibria are ensured once:\footnote{That is, if (\ref{0}) holds, then (\ref{2}) never holds
with equality for real-valued $\theta^{*}$'s. Put differently, the
polynomial which characterizes those values $\theta^*$, for which
(\ref{2}) holds with equality, has two complex roots.}
\begin{eqnarray} ln\Big(\frac{1}{\sqrt{2\pi}}
\frac{1}{\sqrt{(\frac{\sqrt{\alpha_x}}{\alpha_p})^2+\sigma^2_{\mu}}}\Big)\leq0
\quad\Leftrightarrow\quad
\sqrt{\Big(\frac{\sqrt{\alpha_x}}{\alpha_p}\Big)^2+\sigma^2_{\mu}}\geq\frac{1}{\sqrt{2\pi}}.\label{0}
\end{eqnarray}
\end{proof}

In one interpretation, the role of the prior's dispersion
$\sigma_{\mu}$ in Condition (\ref{0}) indicates that polarized
economies, with two large groups (tails) which believe that the
status quo is going to be maintained (abandoned), have unique
equilibria. Put differently, increases in the prior's dispersion
$\sigma_{\mu}$ ``thin out" the group of agents, around the mean
$E[\mu]$ who hold ``moderate" beliefs. And it is this group which
can potentially coordinate on multiple equilibria. Regarding the
prior's weight, $\alpha_p$, we find that, for every given prior
expectation $\mu$, increases in $\alpha_p$ make actions more
predictable. This allows agents to coordinate, as in the common
prior economy of \citet{Mor04}, which contributes towards
equilibrium multiplicity. Comparison of the two \emph{origins} of
belief heterogeneity indicates that increased dispersion of
private signals contributes towards equilibrium multiplicity,
while increases in the dispersion of prior beliefs $\mu$
contributes towards uniqueness. In the next section, we show that
the thinning-out effect is robust to various changes in the
informational environment. And it is the only avenue that can
ensure unique equilibria in all three public signal environments
that follow.

From a technical perspective, we note that Proposition \ref{p1}
has two corollaries:

\begin{coro} If the prior $\mu$ is sufficiently dispersed, such that $\sigma_{\mu}>\frac{1}{\sqrt{2\pi}}$,
equilibria are unique irrespective of the relative precision
$\frac{\sqrt{\alpha_x}}{\alpha_p}$ of private signal and prior.
And in the case where the private signal is uninformative,
$\alpha_x=0$, unique (multiple) equilibria exist if
$\sigma_{\mu}>(<)\frac{1}{\sqrt{2\pi}}$.
\end{coro}
\begin{coro} The uniqueness condition $\sqrt{(\frac{\sqrt{\alpha_x}}{\alpha_p})^2+\sigma^2_{\mu}}\geq\frac{1}{\sqrt{2\pi}}$ converges smoothly to
the uniqueness condition,
$\frac{\sqrt{\alpha_x}}{\alpha_p}\geq\frac{1}{\sqrt{2\pi}}$, of
the \citet{Mor04} common prior game, as
$\sigma_{\mu}\rightarrow0$.
\end{coro}

\section{Public Signals}\label{s45}

To reflect on the role of prior heterogeneity, we introduce three
different types of public signals into our baseline setting. Each
of these public signals is chosen to isolate particular
differences between heterogeneity in priors and heterogeneity in
private signals. The main finding in this section is that only
heterogeneity in priors selects unique equilibria reliably. That
is, if player's disagreement, as measured by $\sigma_{\mu}$, is
substantial, then they play unique equilibria irrespective of the
particular public signal context. That is, the result from the
previous section, i.e., that prior heterogeneity induces
equilibrium uniqueness through the thinning-out effect, is robust
to the introduction of public signals. On the contrary, the role
the private signal will change from environment to environment.

The first public signal, which we introduce into our baseline
model from Section \ref{s1}, is of exogenous quality. That is, it
reveals the true fundamental of the game with exogenous precision
$\alpha_z$. In such an extended setting, we find that the
comparative statics of the subjective prior's precision $\alpha_p$
change: contrary to the baseline setting, where increases in the
prior weight always contribute towards multiplicity, we find that
increases in the prior weight can now shift the modified economy
from multiplicity towards uniqueness. In a second step, we
endogenize the quality of the public signal. To do so, we embed
our coordination game into a financial markets context, where a
stock price aggregates and reveals dispersed private information.
In this setting, we find that the role of private information,
with respect to equilibrium multiplicity, is reversed. Namely,
equilibrium multiplicity is ensured in the limit where private
information becomes arbitrarily precise. On the contrary, the role
of prior dispersion is robust to such changes in the model
structure.

In the last Section \ref{sec1}, we introduce a public signal which
partially reveals the aggregate attack $A$. This signal provides
an environment where changes in the prior's dispersion may, for
intermediate values, induce multiplicity rather than uniqueness.
However, in the limit where prior dispersion grows large, it still
ensures unique equilibria. Finally, as a byproduct of our analysis
in Section \ref{sec1}, we find that public signals in themselves
have an ambiguous effect on equilibrium multiplicity: sufficiently
precise public signals can ensure unique threshold equilibria.
This finding is of independent interest in comparison with the
games of \citet{Mor01,Mor04}, \citet{Hel02}, \citet{Met02},
\citet{Ang06}, where increases in the public signal's precision
unambiguously induce multiple rather than unique threshold
equilibria.

\paragraph{Public Signal with Exogenous Precision}
The public signal
\begin{eqnarray} Z=\theta+\sigma_z\varepsilon,\quad \varepsilon\sim\mathcal{N}(0,1),\label{3}\end{eqnarray}
allows agents to forecast the true state of the fundamental with
precision $\alpha_z=\frac{1}{\sigma_z^2}$. Agents can therefore
use $Z$, in addition to $x$ and $\mu$, to calculate the
probability with which the aggregate attack overwhelms the status
quo. In Appendix \ref{a3}, we show that, if this signal is used as
an additional source of information in the coordination game of
Section \ref{s1}, we have:
\begin{prop}\label{p2} The equilibrium in the public and private information game with heterogenous priors is unique if
\begin{eqnarray}\sqrt{\Big(\frac{\sqrt{\alpha_x}}{\alpha_p+\alpha_z}\Big)^2
+\sigma^2_{\mu}\frac{1}{(1+\frac{\alpha_z}{\alpha_p})^2}}\geq\frac{1}{\sqrt{2\pi}}.\label{4}\end{eqnarray}
In particular, if the prior's dispersion is large, such that
$\sigma_{\mu}\frac{1}{1+\frac{\alpha_z}{\alpha_p}}>\frac{1}{\sqrt{2\pi}}$,
the equilibrium is unique independently of the private signal's
precision $\alpha_x$.
\end{prop}

Compared to the uniqueness condition (\ref{0}) from the baseline
model, we find once again that the modified condition (\ref{4})
has two elements. The first,
$\frac{\sqrt{\alpha_x}}{\alpha_p+\alpha_z}$, reflects the
trade-off between private information $\alpha_x$ and prior
$\alpha_p$, described by \citet{Mor01,Mor04} in an economy without
public signals, or respectively, the trade-off between private
information and public signals $\alpha_z$ which was emphasized by
\citet{Met02} and \citet{Hel02} in an economy with a uniform
uninformative prior. Regarding this first term, we find that
public information and prior are perfect substitutes, and both
contribute to equilibrium multiplicity. The second term
$\sigma^2_{\mu}\frac{1}{(1+\frac{\alpha_z}{\alpha_p})^2}$,
however, shows that increases in the public signal's precision
$\alpha_z$ reduce the effect of the prior's dispersion, while the
prior weight $\alpha_p$ increases it. Condition (\ref{4})
therefore shows that the public signal's precision unambiguously
contributes towards multiplicity. Increases in the prior's weight
$\alpha_p$ on the contrary have an ambiguous consequences as they
shift the economy towards uniqueness (multiplicity) if
$\sigma^2_{\mu}>(<)\frac{\alpha_x}{\alpha_p\alpha_z}$. Finally,
(\ref{4}) reflects that equilibria are unique in the private
information limit where $\alpha_x\rightarrow\infty$.




\paragraph{Public Signal with Endogenous Precision}
To highlight the different implications of prior dispersion and
the dispersion of private signals, we discuss an environment where
the global game is embedded in a financial market setting.
Following, \citet{Atk01}, \citet{Ang06}, and \citet{Hel06}, we
introduce a financial market which aggregates dispersed private
information on the unknown fundamental $\theta$, through its
publicly observable stock price, as in \citet{Gro76,Gro80}, and
\citet{Hel80}.\footnote{That is, we assume that agents trade
stocks prior to the coordination game. These stocks are traded at
a market price $P$ and pay an unknown amount $\theta$. This market
price will, in equilibrium, aggregate dispersed private
information and reveal the true fundamental $\theta$ partially.
Where the partial revelation is due to aggregate noise-trader
activity, $\sigma_{\varepsilon}\varepsilon$,
$\varepsilon\sim\mathcal{N}(0,1)$, on the asset's supply side.} In
one interpretation, the extended model may describe a situation,
where bond investors use a firm's stock price to infer its default
probability, which is of importance for a coordination game that
concerns a potential run on the firm's debt. We show in Appendix
\ref{a2} that it is possible to specify the financial market such
that the public stock price signal, $Z$, partially reveals the
true fundamental $\theta$:
\begin{eqnarray} Z=\theta-\gamma\sigma_{\varepsilon}\sigma^2_{x}\varepsilon,\quad \varepsilon\sim\mathcal{N}(0,1).\label{5}\end{eqnarray}
Thus, the signal's precision
$\alpha_z=\frac{1}{(\gamma\sigma_{\varepsilon})^2}\alpha^2_x$ is
an increasing function of the private signal's precision
$\alpha_x=\frac{1}{\sigma_x^2}$. That is, the stock price's
informativeness increases once the stock investors' information
becomes more informative. In the current context, it is important
that the precision with which this financial market publicly
reveals the true state of the world $\theta$ is increasing faster
(in the private signal's precision) than the private signal's
precision $\alpha_x$ itself. To perform the equilibrium analysis
which concerns the coordination game, we recall (\ref{4}) and note
that $\alpha_z=\alpha_z(\alpha_x)$. This yields
\begin{prop}\label{p3}The equilibrium is unique if
\begin{eqnarray}\sqrt{\Big(\frac{\sqrt{\alpha_x}}{\alpha_p+\alpha_z(\alpha_x)}\Big)^2
+\sigma^2_{\mu}\frac{1}{(1+\frac{\alpha_z(\alpha_x)}{\alpha_p})^2}}\geq\frac{1}{\sqrt{2\pi}},\quad
\alpha_z:=\frac{1}{(\gamma\sigma_{\varepsilon})^2}\alpha^2_x.\label{6}\end{eqnarray}
Multiple equilibria exist in the private information limit where
$\alpha_x\rightarrow\infty$. The equilibrium is unique in the
limit where $\sigma_{\mu}\rightarrow\infty$.
\end{prop}
\begin{proof}Follows from (\ref{6}) with $\alpha_z(\alpha_x)=\frac{1}{(\gamma\sigma_{\varepsilon})^2}\alpha^2_x$.\end{proof}

Proposition \ref{p3} establishes that the finding of \citet{Ang06}
carries over to an economy with heterogenous priors. Namely, if
stock prices aggregate private information rapidly as in
(\ref{5}), then it is precise private information which ensures
equilibrium multiplicity. Moreover, as $\alpha_z(\alpha_x)$
becomes large, it marginalizes the influence of prior
heterogeneity
$\sigma^2_{\mu}\frac{1}{(1+\frac{\alpha_z}{\alpha_p})^2}$.
Finally, if private noise becomes large as $\alpha_x\rightarrow0$
and thus $\alpha_z(\alpha_x)\rightarrow0$, equilibrium
multiplicity depends on prior dispersion, $\sigma_{\mu}$,
alone.\footnote{This finding naturally differs from \citet{Ang06},
where $\sigma_{\mu}=0$, such that noisy private signals
unambiguously induce unique equilibria when public information is
endogenous once $\alpha_x\rightarrow0$ and
$\alpha_z(\alpha_x)\rightarrow0$. Related to this observation, we
note that among all parameters, $\alpha_x,\alpha_p,\sigma_{\mu}$,
the prior's dispersion $\sigma_{\mu}$ is the only parameter which
can ensure equilibrium uniqueness regardless of the values of the
remaining parameters.}

Concerning the different implications of heterogenous priors and
heterogenous private signals, the key insight is that the
endogeneity of public information \emph{inverts} the original
findings of \citet{Mor01,Mor04}, \citet{Met02}, and \citet{Hel02},
where increases in private information induce equilibrium
uniqueness as in (\ref{4}), where the public signal's precision is
exogenous. The same is not true for the role of prior dispersion,
which is, contrary to the private signal's dispersion, robust to
the introduction of an endogenous public price signal and
unambiguously contributes towards equilibrium uniqueness.
\emph{That is, heterogeneity in priors, unlike heterogeneity in
private information, provides a robust anchor for unique
equilibria.}

\subsection{Observing Other's Actions}\label{sec1}

In this section, agents can observe the size of the aggregate
attack through a noisy public signal
$S=\Phi^{-1}(A)+\sigma_{\varepsilon}\varepsilon$ where
$\varepsilon\sim\mathcal{N}(0,1)$. While games where agents can
observe each other's actions were already studied by
\citet{Min03}, the current signal specification is taken from
\citet{Das07} and \citet{Ang06}, since it allows to preserve
normal distributions. The following discussion of equilibrium
uniqueness is accordingly parallel to that in \citet{Ang06}; and
their results obtain as special cases where priors are
uninformative, i.e., where $\alpha_p=0$. Compared to the analysis
in \citet{Ang06} we note that public signals of high precision can
ensure unique threshold equilibria in our specification \emph{if
priors are informative}. That is, the role of the public signal in
\citet{Ang06}, pp. 1733-1734, depends critically on the absence of
an informative prior.\footnote{More precisely, for a priorless
game, \citet{Ang06} show that threshold equilibria are always
unique, but there may exist multiple equilibria in ``strategies".
In the present model, which includes informative (possibly unique)
priors, we show that multiple threshold equilibria may exist.
However, if the public signal, over others' actions $A$, is
sufficiently precise then threshold equilibria are always unique.
The observation that public signals of high quality can ensure
unique rather than multiple threshold equilibria is of interest in
comparison with \citet{Mor01,Mor04}, \citet{Met02}, and
\citet{Hel02}, who show that multiple threshold equilibria emerge
once public signals are of high quality.}

Compared to the previous two signal environments, signal $S$
carries \emph{two types} of information. First, similar to signals
(\ref{3}) and (\ref{5}), the signal $S$ allows agents to make
inference on the true fundamental $\theta$ since
$A=A(\theta,\psi^*)$. Second, unlike signals (\ref{3}) and
(\ref{5}), the particular signal realization $S$ is endogenous in
the sense that $S$ is implicitly defined by
$S=\Phi^{-1}(A(\theta,\psi^*(S)))+\sigma_{\varepsilon}\varepsilon$.
And there may exist multiple signal values $S$ for every given
pair $\theta,\varepsilon$. We examine these potential sources of
multiplicity, namely \emph{equilibrium multiplicity in thresholds
$\theta^*,\psi^*$} and \emph{equilibrium multiplicity in
strategies $S,\psi^*(S)$} in separate steps.

\paragraph{Equilibrium Multiplicity in Thresholds}

In the augmented game, with heterogeneous priors, where agents
observe the aggregate attack through signal $S$, we have:
\begin{eqnarray}
&&S=\Phi^{-1}(A)+\sigma_{\varepsilon}\varepsilon,
\quad \varepsilon\sim\mathcal{N}(0,1) \label{g2123}\\
&&A\equiv P(\psi\leq\psi^{*}(S)|\theta^*)=\theta^*\\
&&P(\theta\leq\theta^{*}|\psi^*,S)=c.\label{g22}\end{eqnarray}
Evaluation of (\ref{g2123})-(\ref{g22}) yields:
\begin{prop}\label{p4}For every given signal realization S, threshold equilibria $\theta^*,\psi^*$ are unique if
$\frac{1}{\sqrt{2\pi}}\leq\frac{(\alpha_x+\alpha_z)}{\alpha_x}\sqrt{\Big(\frac{\sqrt{\alpha_x}}{\alpha_p}\Big)^2+\sigma_{\mu}^2}$
where
$\alpha_z=\frac{\alpha_x^2\alpha_{\psi}}{\sigma^2_{\varepsilon}\alpha^2}=\frac{\alpha_x^2}{\sigma_{\varepsilon}^2(\alpha_x+\alpha_p^2\sigma_{\mu}^2)}$.
And equilibria are unique when priors are either sufficiently
dispersed or when the public signal is sufficiently precise.
\end{prop}
\begin{proof} See Appendix \ref{a4}  \end{proof}
With regard to the role of prior dispersion, Proposition \ref{p4}
shows that prior dispersion contributes towards equilibrium
uniqueness in the generalized setting where agents can observe
each other's actions.\footnote{Note that increases in the prior's
dispersion reduce the public signal's precision
$\alpha_z=\frac{\alpha_x^2\alpha_{\psi}}{\sigma^2_{\varepsilon}\alpha^2}=\frac{\alpha_x^2}{\sigma_{\varepsilon}^2(\alpha_x+\alpha_p^2\sigma_{\mu}^2)}$;
for intermediate values of $\sigma_{\mu}$, it is therefore not
necessarily true that increases in $\sigma_{\mu}$ contribute
towards uniqueness.} The more significant finding, however, is
that the public signal's precision induces equilibrium uniqueness
rather than multiplicity. That is, in the present framework, we
find that the public signal allows agents to coordinate on
\emph{one particular equilibrium} rather than multiple equilibria
as in \citet{Mor01,Mor04}, \citet{Met02}, \citet{Hel02}, and
\citet{Ang06}. Moreover, the comparative statics with regard to
the public signal carry over to the unique prior economy where
$\sigma_{\mu}=0$. Finally, for an uninformative prior where
$\alpha_p=0$, we find that the uniqueness result of \citet{Ang06}
obtains as a special case.\footnote{In a setting with a uniform
uninformative prior, \citet{Ang06}, pp. 1733-1734, prove that
threshold equilibria are always unique irrespective of the
precisions $\alpha_z$ and $\alpha_x$. To obtain this result, one
can either repeat the calculations in Appendix \ref{a4} with
$\alpha_p=0$. Alternatively, one can observe that the uniqueness
condition in Proposition \ref{p4} is always satisfied once a
sufficiently small value $\alpha_p$ is chosen.}




\paragraph{Equilibrium Multiplicity in Strategies}
In this paragraph, we study the uniqueness of the equilibrium with
respect to the signal $S$. In Appendix \ref{a4}, we show that
signal $S$ is, in equilibrium, equivalent to a signal
$Z(S)=\frac{\alpha}{\alpha_x}\psi^*(S)-\frac{\alpha}{\alpha_x}\frac{1}{\sqrt{\alpha_{\psi}}}S=\theta+\frac{\alpha_p}{\alpha_x}E[\mu]
-\sigma_{\varepsilon}\frac{\alpha}{\alpha_x}\frac{1}{\sqrt{\alpha_{\psi}}}\varepsilon$.
And multiple equilibria can emerge in the sense\footnote{At this
point we do not discuss the conceptual validity of this
alternative type/source of equilibrium multiplicity, which implies
that the particular signal S is endogenous, i.e., depends on the
$\psi^*$ chosen. \citet{Ang06}, p.1730, provide a brief discussion
and further references regarding this fundamental problem.} that
there may exist several signal values $S$, and thus several values
$\psi^*(S)$, which satisfy $Z(S)=\bar{Z}$. Concerning this
potential source of equilibrium multiplicity we note


\begin{prop} Equilibria in strategies $\psi^*(S)$ are unique if
$\frac{1}{\sqrt{2\pi}}\leq\sqrt{\Big(\frac{\sqrt{\alpha_x}}{\alpha_p+\alpha_z}\Big)^2+\frac{\sigma_{\mu}^2}{1+\frac{\alpha_z}{\alpha_p}}}$,
with
$\alpha_z=\frac{\alpha_x^2}{\sigma_{\varepsilon}^2(\alpha_x+\alpha_p^2\sigma_{\mu}^2)}$.
In the limit, where $\sigma_{\mu}\rightarrow\infty$ there exists a
unique equilibrium. In the limit where
$\alpha_z\rightarrow\infty$, there exist multiple equilibria in
strategies.\label{p5}\end{prop}
\begin{proof} See Appendix \ref{a4} \end{proof}
Comparison of propositions \ref{p4} and \ref{p5} with regard to
the prior's dispersion yields an important corollary:
\begin{coro} The overall equilibrium is unique in the limit where $\sigma_{\mu}\rightarrow\infty$.
\label{c9}\end{coro} Corollary \ref{c9} underscores the main
result of the paper, i.e., it confirms that sufficiently dispersed
priors ensure unique equilibria. The public signal's precision has
a more differentiated influence on equilibria: it ensures
uniqueness in thresholds if it is sufficiently precise, but at the
same time it opens the door to multiple equilibria in strategies.

\section{Conclusion}\label{s5}

We have introduced heterogenous priors into the canonic global
games model of \citet{Mor01,Mor04}, \citet{Met02}, and
\citet{Hel02}. The analysis of the baseline model indicates that
heterogeneity in priors, unlike heterogeneity in private signals,
makes it more difficult for agents to coordinate on  multiple
equilibria. That is, the \emph{origins} of belief heterogeneity
are of crucial importance to the global games approach:
\emph{heterogeneity in beliefs, which originates from the variance
$\sigma_{\mu}$ in prior expectations, contributes to equilibrium
uniqueness. Dispersion in private signals, on the contrary,
induces equilibrium multiplicity.} In general, the prior's
dispersion can ensure unique equilibria as it ``thins-out" the
group of agents who hold ``moderate" beliefs. That is, it reduces
the mass of agents with moderate beliefs, and it is this group
which can potentially coordinate on multiple equilibria.
Equivalently, our results indicate that if player's disagreement,
as measured by $\sigma_{\mu}$, is substantial, then they play
unique equilibria.

More precisely, we find that if prior dispersion is small,
$(\sigma_{\mu}<\frac{1}{\sqrt{2\pi}})$, equilibrium multiplicity
depends on the relative precisions of private information
$(\alpha_x)$ and the subjective prior $(\alpha_p)$. If priors are
sufficiently dispersed, $(\sigma_{\mu}>\frac{1}{\sqrt{2\pi}})$,
equilibria are unique \emph{irrespective} of the relative weights
that players assign to private signals and priors. If prior
dispersion $(\sigma_{\mu})$ vanishes, the original uniqueness
condition for the common prior economy
$(\frac{\sqrt{\alpha_x}}{\alpha_p}\geq\frac{1}{\sqrt{2\pi}})$
obtains.

To compare the implications of prior dispersion and dispersion in
private information, we have discussed a modified game in which a
financial market aggregates private information into a public
price signal. Such a modified environment \emph{inverts} the
original findings of \citet{Mor01,Mor04}, \citet{Met02}, and
\citet{Hel02}: increases in private information now induce
equilibrium multiplicity instead of uniqueness. The same is not
true for the role of prior dispersion, which is robust to such a
change in the modelling environment and contributes unambiguously
towards equilibrium uniqueness. Put differently, the extended
model indicates that prior dispersion, rather than arbitrarily
precise private information, anchors unique equilibria reliably.

In general, we found that sufficiently dispersed priors ensure
unique equilibria across all three public signal environments.
Regarding these public signals, it turned out that their
implications in themselves varied significantly from case to case:
increases in the public signal's precision ensure multiple
threshold equilibria in the first two environments, where signals
only contain information on the unknown fundamental. The opposite
can be true in the third environment, where public signals allow
agents to observe each other's actions. If such signals are of
high quality, they can enable agents to coordinate on \emph{one
unique threshold equilibrium}.

Unlike previous studies, which have introduced heterogenous priors
into the global games framework, we have given explicit conditions
in terms of means and variances, which allow to study equilibrium
multiplicity for a large economy at \emph{and away from the
private information limit}. That is, the present framework
facilitates comparative statics in the information structure
itself, which allows to characterize and compare the different
implications of belief heterogeneity which originate from priors
and private signals, respectively. Moreover, these comparative
statics are useful in those applications of the global games
framework where it is interesting, or necessary, to study the
interaction of private and public information away from the limit
where private signals are infinitely precise.


\newpage

\begin{appendix}

\section{Game with Exogenous Public Information}\label{a3}

In this appendix, we derive the uniqueness condition that obtains
once we augment our baseline model of Section \ref{s1} with a
public signal:
\begin{eqnarray} Z=\theta+\sigma_z\varepsilon,\quad \varepsilon\sim\mathcal{N}(0,1).  \end{eqnarray}
This signal allows agents to improve their forecast of the
probability with which the aggregate attack overwhelms the status
quo. The modified payoff indifference condition therefore
reads:\begin{eqnarray}
P(\theta\leq\theta^*|x^*,\mu,Z)=c\quad\Leftrightarrow \quad
\Phi\Big(\sqrt{\alpha}(\theta^*-\frac{\alpha_x}{\alpha}x^*-\frac{\alpha_p}{\alpha}\mu-\frac{\alpha_z}{\alpha}Z)\Big)=c;
\quad \alpha=\alpha_x+\alpha_p+\alpha_z,
\label{pic2}\end{eqnarray} where $\Phi()$ is the cumulative of the
standard normal distribution. Again, we define
$\psi\equiv\frac{\alpha_x}{\alpha}x+\frac{\alpha_p}{\alpha}\mu$
and rewrite (\ref{pic2}) as:
\begin{eqnarray}
\Phi\Big(\sqrt{\alpha}(\theta^*-\psi^*-\frac{\alpha_z}{\alpha}Z)\Big)=c;
\quad \alpha=\alpha_x+\alpha_p+\alpha_z.
\label{pic3}\end{eqnarray} The PIC in (\ref{pic3}) locates a
critical $\psi^*$ such that agents attack if $\psi>\psi^*$ and do
not attack if $\psi<\psi^*$. To calculate the mass of attacking
agents, we note that
$\psi|\theta\sim\mathcal{N}(\frac{\alpha_x}{\alpha}\theta+\frac{\alpha_p}{\alpha}E[\mu],(\frac{\alpha^2}{\alpha_x+\alpha^2_p\sigma^2_{\mu}})^{-1})$.
Once we define
$\alpha_{\psi}\equiv\frac{\alpha^2}{\alpha_x+\alpha^2_p\sigma^2_{\mu}}$,
the CMC can be written as:
\begin{eqnarray} P(\psi<\psi^*|\theta^*)=\theta^*\quad
\Leftrightarrow\quad
\Phi\Big(\sqrt{\alpha_{\psi}}(\psi^*-\frac{\alpha_x}{\alpha}\theta^*-\frac{\alpha_p}{\alpha}E[\mu])\Big)=\theta^*.\label{cmc3}
\end{eqnarray}
Substitution of (\ref{pic3}) into (\ref{cmc3}) again yields a
one-dimensional equation in the threshold level $\theta^*$:
\begin{eqnarray}
\Phi\Big(\sqrt{\alpha_{\psi}}(\frac{\alpha_z+\alpha_p}{\alpha}\theta^*
-\frac{\alpha_z}{\alpha}Z-\frac{\alpha_{p}}{\alpha}E[\mu]-\Phi^{-1}(c)\frac{1}{\sqrt{\alpha}})\Big)=\theta^*.\end{eqnarray}
Accordingly, equilibria are unique if:\footnote{Recall that
$\alpha_{\psi}\equiv\frac{\alpha^2}{\alpha_x+\alpha^2_p\sigma^2_{\mu}}$.}
\begin{eqnarray}\sqrt{\alpha_{\psi}}\frac{\alpha_z+\alpha_p}{\alpha}\frac{1}{\sqrt{2\pi}}\leqq1\quad
\Leftrightarrow\quad
\sqrt{\Big(\frac{\sqrt{\alpha_x}}{\alpha_p+\alpha_z}\Big)^2+\sigma^2_{\mu}\frac{1}{(1+\frac{\alpha_z}{\alpha_p})^2}}\geq\frac{1}{\sqrt{2\pi}}\end{eqnarray}
which is what we needed to show.

\section{Financial Market and Information Aggregation}\label{a2}

In this appendix, we present a financial market that aggregates
and reveals dispersed private information, on the fundamental
$\theta$, through the stock price. For the present purpose, it is
convenient to pick a special case of the linear CARA-normal noise
trader equilibrium discussed by \citet{Gro76,Gro80},
\citet{Hel80}, and \citet{Ang06}. Stocks are traded at a market
price $P$ and pay an unknown amount $\theta$, which represents the
firm's fundamental strength. This market price will, in
equilibrium, aggregate dispersed private information and reveal
the true fundamental $\theta$ partially. Where the partial
revelation is due to aggregate noise-trader activity,
$\sigma_{\varepsilon}\varepsilon$,
$\varepsilon\sim\mathcal{N}(0,1)$, on the asset's supply side. To
characterize the market price signal, we proceed in three steps.
First, we guess that there exists a linear price function,
$P=\eta_1\theta+\eta_2\varepsilon+c$. Regarding $\theta$, this
function is informationally equivalent to a signal
$Z\equiv\frac{P-c}{\eta_1}=\theta+\frac{\eta_2}{\eta_1}\varepsilon$,
which reveals the true fundamental with precision
$\alpha_z=\frac{\eta^2_1}{\eta^2_2}$. Second, given this price
function, we characterize individual demands based on the
information $x,\mu,Z$, and calculate the market equilibrium.
Finally, we determine the ratio $\frac{\eta^2_1}{\eta^2_2}$ as
$\alpha_x^2\frac{1}{\gamma^2\sigma^2_{\varepsilon}}$. That is,
price signal $Z$ indeed carries information
$\alpha_z=\alpha_x^2\frac{1}{\gamma^2\sigma^2_{\varepsilon}}$ as
claimed in (\ref{5}) in the main text.

\paragraph{Demand} Agents choose their optimal demands
$k_i$ for the asset to maximize expected CARA utility:
\begin{eqnarray}
k_{i}&&=\underset{k_{i}}{\arg\max}\{\mathbb{E}[-e^{-\gamma(\theta-P)k_i}|x_i,\mu,Z]\}\nonumber\\
&&=\underset{k_{i}}{\arg\max}\{
\gamma\mathbb{E}[(\theta-P)k_{i}|x_i,\mu,Z]-\frac{\gamma^{2}}{2}Var[(\theta-P)k_{i}|x_i,\mu,Z]\}\nonumber\\
&&=\underset{k_{i}}{\arg\max}\{
\gamma(\frac{\alpha_x}{\alpha}x_i+\frac{\alpha_p}{\alpha}\mu+\frac{\alpha_z}{\alpha}Z-P)k_{i}
-\frac{\gamma^{2}}{2}k_{i}^{2}\frac{1}{\alpha}\},\quad\alpha=\alpha_x+\alpha_p+\alpha_{z}\nonumber
\end{eqnarray}
and the individual demand function writes:
\begin{eqnarray}
k_{i}^{d}=\frac{\frac{\alpha_x}{\alpha}x_i+\frac{\alpha_p}{\alpha}\mu+\frac{\alpha_z}{\alpha}Z-P}
{\alpha^{-1}\gamma}.\label{ap21}
\end{eqnarray}

\paragraph{Equilibrium}
Aggregate supply $K^S=\sigma_{\varepsilon}\varepsilon$, is
unobservable and distorted by noise-trader activity
$\varepsilon\sim\mathcal{N}(0,1)$. From (\ref{ap21}), we find that
aggregate demand $K^D$ is:
\begin{eqnarray} K^D=\int_{[0,1]}\int k_i(\mu)\phi(\mu)d\mu di
=\frac{\frac{\alpha_x}{\alpha}\theta+\frac{\alpha_p}{\alpha}E[\mu]+\frac{\alpha_z}{\alpha}Z-P}
{\alpha^{-1}\gamma}.  \end{eqnarray} Equilibrium requires that:
\begin{eqnarray} K^D=K^S\quad\Leftrightarrow\quad
\frac{\frac{\alpha_x}{\alpha}\theta+\frac{\alpha_p}{\alpha}E[\mu]+\frac{\alpha_z}{\alpha}Z-P}
{\alpha^{-1}\gamma}=\sigma_{\varepsilon}\varepsilon.\label{ap22}
\end{eqnarray}
To close the argument, we now resubstitute $Z=\frac{P-c}{\eta_1}$
and calculate the ratio $\frac{\eta_1}{\eta_2}$. First, we solve
(\ref{ap22}) for $P$ to obtain:
\begin{eqnarray}P=\frac{\alpha_x}{\alpha-\frac{\alpha_z}{\eta_1}}\theta
-\frac{\gamma\sigma_{\varepsilon}}{\alpha-\frac{\alpha_z}{\eta_1}}\varepsilon
+\frac{\alpha_pE[\mu]-\alpha_z\frac{c}{\eta_1}}{\alpha-\frac{\alpha_z}{\eta_1}}.\label{ap25}
\end{eqnarray}
Comparison of (\ref{ap25}) with our initial guess,
$P=\eta_1\theta+\eta_2\varepsilon+c$, indicates that
$\eta_1,\eta_2$ must satisfy:
\begin{eqnarray}\eta_1=\frac{\alpha_x}{\alpha-\frac{\alpha_z}{\eta_1}},
\quad
\eta_2=-\frac{\gamma\sigma_{\varepsilon}}{\alpha-\frac{\alpha_z}{\eta_1}};\quad\alpha=\alpha_x+\alpha_p+\alpha_{z}.\label{ap22.1}
\end{eqnarray}
We quickly determine $\eta_1=\frac{\alpha_x+\alpha_z}{\alpha}$,
$\eta_2=-\frac{(\alpha_x+\alpha_z)\sqrt{\alpha_z}}{\alpha}$ to
calculate
$\frac{\eta_1}{\eta_2}=-\alpha_x\frac{1}{\gamma\sigma_{\varepsilon}}$.
At the same time, it follows from the definition of Z that
$Z=\frac{P-c}{\eta_1}=\theta+\frac{\eta_2}{\eta_1}\varepsilon$.
Hence, agents who observe $P$, and know the model's coefficients,
receive a signal
$Z=\theta-\alpha_x^{-1}\gamma\sigma_{\varepsilon}\varepsilon$, as
claimed in (\ref{5}) in the main text.


\section{Proof of Propositions \ref{p4} and \ref{p5}}\label{a4}
In this appendix, we start by laying out the equations that
describe equilibria. In turn, we characterize the possible
equilibria described in propositions \ref{p4} and \ref{p5} in two
separate paragraphs.

We recall the model from the main text
\begin{eqnarray}
&&S=\Phi^{-1}(A)+\sigma_{\varepsilon}\varepsilon,
\quad \varepsilon\sim\mathcal{N}(0,1) \label{g2}\\
&&A\equiv P(\psi\leq\psi^{*}(S)|\theta)=\theta\label{cmc4}\\
&&P(\theta\leq\theta^{*}|x,\mu,S)=c\label{pic4}\end{eqnarray} To
calculate equilibria, we recall that agents act on
$x=\theta+\sigma_x\xi$, with $\xi\sim\mathcal{N}(0,1)$ and
$\theta|\mu\sim\mathcal{N}(\mu,\sigma_p^2)$, where the prior $\mu$
is distributed over the population as
$\mu\sim\mathcal{N}(E[\mu],\sigma_{\mu})$. Moreover, we define
$\psi\equiv\frac{\alpha_x}{\alpha}x+\frac{\alpha_p}{\alpha}\mu$
with $\alpha=\alpha_x+\alpha_p+\alpha_z$. The PIC (\ref{pic4}) now
writes as:
\begin{eqnarray}
\Phi\Big(\sqrt{\alpha}(\theta^*-\psi^*-\frac{\alpha_z}{\alpha}Z)\Big)=c;
\quad \alpha=\alpha_x+\alpha_p+\alpha_z.
\label{pic4.1}\end{eqnarray} Again, (\ref{pic4.1}) defines a
critical $\psi^*(Z)$ such that agents attack if $\psi\leq\psi^*$
and do not attack if $\psi>\psi^*$. To calculate the mass of
attacking agents, we note that
$\psi|\theta\sim\mathcal{N}(\frac{\alpha_x}{\alpha}\theta+\frac{\alpha_p}{\alpha}E[\mu],(\frac{\alpha^2}{\alpha_x+\alpha^2_p\sigma^2_{\mu}})^{-1})$.
Once we define
$\alpha_{\psi}\equiv\frac{\alpha^2}{\alpha_x+\alpha^2_p\sigma^2_{\mu}}$,
the CMC (\ref{cmc4}) can be written as:
\begin{eqnarray} A=P(\psi<\psi^*|\theta^*)=\theta^*\quad
\Leftrightarrow\quad
\Phi\Big(\sqrt{\alpha_{\psi}}(\psi^*-\frac{\alpha_x}{\alpha}\theta^*-\frac{\alpha_p}{\alpha}E[\mu])\Big)=\theta^*.\label{cmc4.2}
\end{eqnarray}
Using this expression for the aggregate attack $A$ in condition
(\ref{cmc4.2}), we can return to the public signal $S$ in
(\ref{g2}) and write:
\begin{eqnarray} S=\sqrt{\alpha_{\psi}}(\psi^*-\frac{\alpha_x}{\alpha}\theta-\frac{\alpha_p}{\alpha}E[\mu])+\sigma_{\varepsilon}\varepsilon.
\label{s10}\end{eqnarray} Where $S$ in (\ref{s10}) is
informationally equivalent to a signal
\begin{eqnarray}Z(S)\equiv\frac{\alpha}{\alpha_x}\psi^*(S)-\frac{\alpha}{\alpha_x}\frac{1}{\sqrt{\alpha_{\psi}}}S=\theta+\frac{\alpha_p}{\alpha_x}E[\mu]
-\sigma_{\varepsilon}\frac{\alpha}{\alpha_x}\frac{1}{\sqrt{\alpha_{\psi}}}\varepsilon.\label{s11}\end{eqnarray}
Regarding (\ref{s11}), we note that $Z(S)$ contains two aspects
(i) $Z$ is a noisy public signal which reveals the true state of
the economy $\theta$ with precision
$\alpha_z=\frac{\alpha_x^2\alpha_{\psi}}{\sigma^2_{\varepsilon}\alpha^2}=\frac{\alpha_x^2}{\sigma_{\varepsilon}^2(\alpha_x+\alpha_p^2\sigma_{\mu}^2)}$
and (ii) the signal $S$ allows agents to align their strategies
$\psi^*(S)$. That is, for every given $\bar{Z}$, there may be
several $S$ such that $Z(S)=\bar{Z}$. That is, there is a
potential source of equilibrium multiplicity, concerning $S$, to
which we turn in Paragraph 2) of this appendix. For now, we take
$S$ as given and study the threshold equilibria
$\theta^*(S),\psi^*(S)$.

\paragraph{Proof of Proposition \ref{p4}: Multiplicity in Thresholds $\theta^*$}

For every given signal $S$, we rewrite (\ref{cmc4.2}) as:
\begin{eqnarray}
\psi^*=\Phi^{-1}(\theta^*)\frac{1}{\sqrt{\alpha_{\psi}}}+\frac{\alpha_x}{\alpha}\theta^*+\frac{\alpha_p}{\alpha}E[\mu].\label{cmc4.3}
\end{eqnarray}
To obtain an equation in $\theta^*$ only, we substitute
$Z(S)=\frac{\alpha}{\alpha_x}\psi^*(S)-\frac{\alpha}{\alpha_x}\frac{1}{\sqrt{\alpha_{\psi}}}S$
and (\ref{cmc4.3}) into (\ref{pic4.1}). Rearranging then yields:
\begin{eqnarray} \Phi\Big(\sqrt{\alpha}(\frac{\alpha_p}{\alpha}\theta^*
-\frac{1}{\sqrt{\alpha_{\psi}}}\frac{\alpha_x+\alpha_z}{\alpha_x}\Phi^{-1}(\theta^*)
-\frac{(\alpha_x+\alpha_z)\alpha_p}{\alpha_x\alpha}E[\mu]
+\frac{\alpha_z}{\alpha_x}\frac{1}{\sqrt{\alpha_{\psi}}}S)\Big)=c
\label{a11}\end{eqnarray} To derive the uniqueness condition,
which ensures that there exist only one $\theta^*(S)$ for every
given signal $S$, we differentiate (\ref{a11}) with respect to
$\theta^*$:\footnote{Note that for $y=\Phi^{-1}(\theta^{*})$, we
have $\frac{d\theta^{*}}{dy}=\phi(y)$ and thus
$\frac{dy}{d\theta^{*}}=\frac{1}{\phi(y)}=\frac{1}{\phi(\Phi^{-1}(\theta^{*}))}$.}
\begin{eqnarray} \phi(\Phi(\theta^*)^{-1})\leq\frac{(\alpha_x+\alpha_z)}{\alpha_x}\sqrt{\Big(\frac{\sqrt{\alpha_x}}{\alpha_p}\Big)^2+\sigma_{\mu}^2},  \end{eqnarray}
and hence, threshold equilibria are always unique iff
$\frac{1}{\sqrt{2\pi}}\leq\frac{(\alpha_x+\alpha_z)}{\alpha_x}\sqrt{\Big(\frac{\sqrt{\alpha_x}}{\alpha_p}\Big)^2+\sigma_{\mu}^2}$.
Otherwise, if
$\frac{1}{\sqrt{2\pi}}\geq\frac{(\alpha_x+\alpha_z)}{\alpha_x}\sqrt{\Big(\frac{\sqrt{\alpha_x}}{\alpha_p}\Big)^2+\sigma_{\mu}^2}$,
there may exist up to three threshold equilibria\\
$\theta^*_1(S),\psi_1^*(S);\theta^*_2(S),\psi_2^*(S);\theta^*_3(S),\psi_3^*(S)$
for every given signal value $S$.

\paragraph{Proof of Proposition \ref{p5}: Multiplicity in Strategies $\psi^*(S)$}
To preclude multiple solutions\footnote{The existence of at least
one solution is ensured. It follows from (\ref{cmc4.3}) that
$lim_{S\rightarrow\infty}\frac{\Phi(\theta^*(S))^{-1}}{S}$ and
$lim_{S\rightarrow-\infty}\frac{\Phi(\theta^*(S))^{-1}}{S}$ are
constants. Rewriting
$Z(S)=\frac{\alpha}{\alpha_x}\psi^*(S)-\frac{\alpha}{\alpha_x}\frac{1}{\sqrt{\alpha_{\psi}}}S$
as
$Z(S)=\frac{\alpha}{\alpha_x}S(\frac{\psi^*(S)}{S}-\frac{\alpha}{\alpha_x}\frac{1}{\sqrt{\alpha_{\psi}}})$
and recalling $\psi(\theta^*(S))$ as given in (\ref{cmc4.3}) one
can show that $Z(S)$ varies with $S$ between $\infty$ and
$-\infty$.} $S(\bar{Z})$ to the equation $Z(S)=\bar{Z}$, where
$Z(S)=\frac{\alpha}{\alpha_x}\psi^*(S)-\frac{\alpha}{\alpha_x}\frac{1}{\sqrt{\alpha_{\psi}}}S$,
it will suffice to show that $\frac{\partial Z(S)}{\partial
S}_{|(\ref{s11})}=\frac{\alpha}{\alpha_x}\frac{\partial\psi^*}{\partial
S}-\frac{\alpha}{\alpha_x}\frac{1}{\sqrt{\alpha_{\psi}}}\leq 0$.
To calculate the derivative
$\frac{\partial\psi(\theta^*(S))}{\partial
S}=\frac{\partial\psi^*}{\partial
\theta^*}\frac{\partial\theta^*}{\partial S}$, defined by
(\ref{cmc4.3}) and (\ref{a11}), we differentiate (\ref{a11}) which
yields $\frac{\partial\theta^*}{\partial
S}=\frac{-\frac{\alpha_z}{\alpha_x}\frac{1}{\sqrt{\alpha_{\psi}}}}{\frac{\alpha_p}{\alpha}
-\frac{1}{\sqrt{\alpha_{\psi}}}\frac{\alpha_x+\alpha_z}{\alpha_x}\frac{1}{\phi(\Phi(\theta^*)^{-1})}}$
and (\ref{cmc4.3}), (which is a $1:1$ mapping between $\psi^*$ and
$\theta^*$), to obtain
$\frac{\partial\psi^*}{\partial\theta^*}=\frac{1}{\sqrt{\alpha_{\psi}}\phi(\Phi(\theta^*)^{-1})}+\frac{\alpha_x}{\alpha}$.
Hence, we have
\begin{eqnarray} \frac{\partial Z(S)}{\partial
S}&&=\frac{\alpha}{\alpha_x}\frac{\partial\psi^*}{\partial
S}-\frac{\alpha}{\alpha_x}\frac{1}{\sqrt{\alpha_{\psi}}}=\frac{\alpha}{\alpha_x}\frac{\partial\psi^*}{\partial
\theta^*}\frac{\partial\theta^*}{\partial
S}-\frac{\alpha}{\alpha_x}\frac{1}{\sqrt{\alpha_{\psi}}}\nonumber\\
&&=\underset{-}{\underbrace{-\frac{\alpha}{\alpha_x}\frac{1}{\sqrt{\alpha_{\psi}}}}}
+\underset{+}{\underbrace{\Big(\frac{1}{\sqrt{\alpha_{\psi}}\phi(\Phi(\theta^*)^{-1})}+\frac{\alpha_x}{\alpha}\Big)}}
\underset{+/-}{\underbrace{\frac{-\frac{\alpha_z}{\alpha_x}\frac{1}{\sqrt{\alpha_{\psi}}}}{\frac{\alpha_p}{\alpha}
-\frac{1}{\sqrt{\alpha_{\psi}}}\frac{\alpha_x+\alpha_z}{\alpha_x}\frac{1}{\phi(\Phi(\theta^*)^{-1})}}}}\frac{\alpha}{\alpha_x}.\label{100}\end{eqnarray}
Once we recall that
$\alpha_{\psi}\equiv\frac{\alpha^2}{\alpha_x+\alpha^2_p\sigma^2_{\mu}}$,
rearranging (\ref{100}) gives:
\begin{eqnarray}
\frac{\partial Z(S)}{\partial
S}&&=-\frac{\sqrt{\alpha_x+\alpha^2_p\sigma^2_{\mu}}}{\alpha_x}+\Big(\sqrt{\alpha_x+\alpha^2_p\sigma^2_{\mu}}
+\alpha_x\phi(\Phi(\theta^*)^{-1})\Big)
\frac{-\frac{\alpha_z}{\alpha_x}}{\frac{\phi(\Phi(\theta^*)^{-1})\alpha_p\alpha_x}{\sqrt{\alpha_x+\alpha^2_p\sigma^2_{\mu}}}-(\alpha_x+\alpha_z)}\nonumber\\
&&=\frac{1}{\alpha_x}\Big[-\sqrt{\alpha_x+\alpha^2_p\sigma^2_{\mu}}+\Big(\sqrt{\alpha_x+\alpha^2_p\sigma^2_{\mu}}
+\alpha_x\phi(\Phi(\theta^*)^{-1})\Big)
\frac{-\alpha_z}{\frac{\phi(\Phi(\theta^*)^{-1})\alpha_p\alpha_x}{\sqrt{\alpha_x+\alpha^2_p\sigma^2_{\mu}}}-(\alpha_x+\alpha_z)}\Big]\nonumber\\
&&=\frac{1}{\alpha_x}\Big[-\alpha_x(\alpha_p+\alpha_z)\phi(\Phi(\theta^*)^{-1})+\alpha_x\sqrt{\alpha_x+\alpha^2_p\sigma^2_{\mu}}\Big]\frac{1}{\frac{\phi(\Phi(\theta^*)^{-1})\alpha_p\alpha_x}{\sqrt{\alpha_x+\alpha^2_p\sigma^2_{\mu}}}-(\alpha_x+\alpha_z)}
\nonumber\\&&=\Big[-(\alpha_p+\alpha_z)\phi(\Phi(\theta^*)^{-1})+\sqrt{\alpha_x+\alpha^2_p\sigma^2_{\mu}}\Big]\frac{1}{\frac{\phi(\Phi(\theta^*)^{-1})\alpha_p\alpha_x}{\sqrt{\alpha_x+\alpha^2_p\sigma^2_{\mu}}}-(\alpha_x+\alpha_z)}
\label{e1}\end{eqnarray} From (\ref{e1}), and the fact that
$\theta^*\in(0,1)$, it follows that equilibria in strategies are
unique if
\begin{eqnarray}\frac{1}{\sqrt{2\pi}}\leq\sqrt{\Big(\frac{\sqrt{\alpha_x}}{\alpha_p+\alpha_z}\Big)^2+\frac{\sigma_{\mu}^2}{1+\frac{\alpha_z}{\alpha_p}}},
\quad
\alpha_z=\frac{\alpha_x^2}{\sigma_{\varepsilon}^2(\alpha_x+\alpha_p^2\sigma_{\mu}^2)}\label{e11}\end{eqnarray}
and
\begin{eqnarray}\frac{1}{\sqrt{2\pi}}\leq\sqrt{\frac{(\alpha_x+\alpha_z)^2}{\alpha_x\alpha_p^2}+\frac{\sigma_{\mu}^2(\alpha_x+\alpha_z)^2}{\alpha_x^2}}.\label{e12}\end{eqnarray}
That is, once (\ref{e11}) and (\ref{e12}) hold, we have
$\frac{\partial Z(S)}{\partial S}<0$, which ensures unique
solutions $S(\bar{Z})$ to the equation $Z(S)=\bar{Z}$. Comparison
shows that inequality (\ref{e12}) is less restrictive than
(\ref{e11}).\footnote{$\sqrt{\frac{\alpha_x}{(\alpha_p+\alpha_z)^2}+\frac{\sigma_{\mu}^2}{1+\frac{\alpha_z}{\alpha_p}}}\leq\sqrt{\frac{(\alpha_x+\alpha_z)^2}{\alpha_x\alpha_p^2}+\frac{\sigma_{\mu}^2(\alpha_x+\alpha_z)^2}{\alpha_x^2}}$
follows from the inequalities
$\frac{\alpha_x}{(\alpha_p+\alpha_z)^2}\leq\frac{(\alpha_x+\alpha_z)^2}{\alpha_x\alpha_p^2}$
and
$\frac{\sigma_{\mu}^2}{1+\frac{\alpha_z}{\alpha_p}}\leq\frac{\sigma_{\mu}^2(\alpha_x+\alpha_z)^2}{\alpha_x^2}$,
which are easy to verify.} Evaluation of (\ref{e11}) therefore
yields:

\begin{enumerate}

\item In the limit where $\sigma_{\mu}\rightarrow\infty$,
equilibria in strategies are unique.

\item  In the limit\footnote{Note that the public signal's
precision is endogenous, i.e., given by
$\alpha_z=\frac{\alpha_x^2}{\sigma_{\varepsilon}^2(\alpha_x+\alpha_p^2\sigma_{\mu}^2)}$.
However, the present observation is informative in the sense that
the public signal's precision can be varied through
$\sigma_{\varepsilon}^2$ which is independent of the other
parameters.} where $\alpha_{z}\rightarrow\infty$, there exist
multiple equilibria in strategies.

\item Multiple equilibria are ensured in the limit where
$\alpha_{x}\rightarrow\infty$.

\item Finally, in the special case where $\sigma^2_{\mu}=0$ and
$\alpha_p=0$, condition (\ref{e11}) collapses into the uniqueness
condition $\sqrt{2\pi}\leq\frac{\sqrt{\alpha_x}}{\alpha_z}$, of
\citet{Ang06} pp. 1733-1734, which is nested in the present
framework.

\end{enumerate}

\section{Alternative Derivation of (\ref{0})}\label{A0}

This appendix contains a derivation of (\ref{0}) which
``explicitly" accounts for the influence which the prior's
distribution has on the critical mass condition. Recalling the
PIC, we have:
\begin{eqnarray} P(\theta\leq\theta^*|x^*,\mu)=c\quad\Leftrightarrow \quad
\Phi\Big(\sqrt{\alpha}(\theta^*-\frac{\alpha_x}{\alpha}x^*-\frac{\alpha_p}{\alpha}\mu)\Big)=c;
\quad \alpha=\alpha_x+\alpha_p. \label{pic}\end{eqnarray}
Regarding the CMC, we now account explicitly for the prior's
distribution and write:
\begin{eqnarray}  A(x^*(\mu),\theta^*)=P(x\leq x^*|\theta^*)=\theta^*\quad\Leftrightarrow \quad
\int_{-\infty}^{\infty}\Phi\Big(\sqrt{\alpha_x}\Big(x^*(\mu)-\theta^*\Big)\Big)\phi(\mu)d\mu=\theta^*,
\label{cmc}\end{eqnarray} where $\Phi()$ represents the cumulative
of the standard normal distribution and $\phi(\mu)$ is the normal
density of the prior. To show that the pairs $x^*(\mu),\theta^*$,
which solve (\ref{pic}) and (\ref{cmc}) are unique if
$\sqrt{(\frac{\sqrt{\alpha_x}}{\alpha_p})^2+\sigma^2_{\mu}}\geq\frac{1}{\sqrt{2\pi}}$,
we solve (\ref{pic}) for the threshold level
$x^*=-\Phi^{-1}(c)\frac{\sqrt{\alpha}}{\alpha_x}+\theta^*\frac{\alpha}{\alpha_x}-\frac{\alpha_p}{\alpha_x}\mu$,
and substitute $x^*$ into (\ref{cmc}) to obtain a one-dimensional
equation in $\theta^*$ alone:
\begin{eqnarray} \int_{-\infty}^{\infty}\Phi\Big(\sqrt{\alpha_x}\Big(-\Phi^{-1}(c)\frac{\sqrt{\alpha}}{\alpha_x}+\theta^*\frac{\alpha}{\alpha_x}-\frac{\alpha_p}{\alpha_x}\mu-\theta^*\Big)\Big)\phi(\mu)d\mu=\theta^*.\nonumber\end{eqnarray}
The sufficient condition for uniqueness of the threshold
$\theta^*$ is therefore:
\begin{eqnarray}
\frac{\alpha_p}{\sqrt{\alpha_x}}\frac{1}{\sqrt{2\pi}}\int_{-\infty}^{\infty}e^{-\Big(\sqrt{\alpha_x}\Big(
-\Phi^{-1}(c)\frac{\sqrt{\alpha}}{\alpha_x}+\theta^*\frac{\alpha_p}{\alpha_x}-\frac{\alpha_p}{\alpha_x}\mu\Big)\Big)^2\frac{1}{2}}\phi(\mu)d\mu\leqq1.
\label{20}\end{eqnarray} To obtain the final condition, we recall
that $\mu\sim\mathcal{N}(E[\mu],\sigma_{E[\mu]})$, and use the
moment generating function for the non-central $\chi^2$
distribution in Paragraph 1) to rewrite (\ref{20}), as:
\begin{eqnarray} \frac{\alpha_p}{\sqrt{\alpha_x}}\frac{1}{\sqrt{2\pi}}\frac{1}{\sqrt{1+(\sigma_{\mu}\frac{\alpha_p}{\sqrt{\alpha_x}})^2}}
e^{-\frac{1}{2}\Big(\sqrt{\alpha_x}\Big(
-\Phi^{-1}(c)\frac{\sqrt{\alpha}}{\alpha_x}+\theta^*\frac{\alpha_p}{\alpha_x}-\frac{\alpha_p}{\alpha_x}E[\mu]\Big)\Big)^2\frac{1}{1+
(\sigma_{\mu}\frac{\alpha_p}{\sqrt{\alpha_x}})^2}}\leq1.\label{ap2}\end{eqnarray}
Taking logarithms yields the uniqueness condition
$\sqrt{(\frac{\sqrt{\alpha_x}}{\alpha_p})^2+\sigma^2_{\mu}}\geq\frac{1}{\sqrt{2\pi}}$.
\paragraph{Moment Generating Function} Regarding (\ref{20}), we
recall our assumption that
$\mu\sim\mathcal{N}(E[\mu],\sigma_{E[\mu]})$. We can therefore
define $y=-\Phi^{-1}(c)
\frac{\sqrt{\alpha}}{\alpha_x}+\theta^*\frac{\alpha_p}{\alpha_x}-\frac{\alpha_p}{\alpha_x}\mu$,
where
$y\sim\mathcal{N}(-\Phi^{-1}(c)\frac{\sqrt{\alpha}}{\alpha_x}+\theta^*\frac{\alpha_p}{\alpha_x}-\frac{\alpha_p}{\alpha_x}E[\mu],\sigma_{\mu}\frac{\alpha_p}{\alpha_x})$.
If a variable $y$ is normally distributed with mean $E[y]$ and
variance $\sigma^2_{y}$, then $z^2=\frac{y^2}{\sigma^2_y}$ is
non-centrally $\chi^2$ distributed. We can therefore use the
moment generating function for the non-central $\chi^2$
distribution (\citet{Rao65}, p.
181):\footnote{\label{f1}$\mathbb{E}e^{tz^{2}}=\frac{1}{\sqrt{2\pi}}\int_{-\infty}^{\infty}
e^{tz^{2}}e^{-\frac{\left(z-E[z]\right)^{2}}{2}}dz=\frac{1}{\sqrt{2\pi}}\int_{-\infty}^{\infty}
e^{-\frac{\left(1-2t\right)z^{2}-2E[z]
z+E[z]^{2}}{2}}dz=\frac{1}{\sqrt{2\pi}}\int_{-\infty}^{\infty}
e^{-\frac{\left(\sqrt{1-2t}z-E[z]\right)^{2}-\frac{E[z]^{2}}{1-2t}+E[z]^{2}}{2}}dz=\frac{1}{\sqrt{1-2t}\sqrt{2\pi}}\int_{-\infty}^{\infty}
e^{-\frac{E[z]^{2}t}{1-2t}}e^{-\frac{\left(\sqrt{1-2t}z-E[z]\right)^{2}}{2}}d\left(\sqrt{1-2t}z\right)=\frac{1}{\sqrt{1-2t}}e^{-\frac{E[z]^{2}t}{1-2t}}$.}
\begin{eqnarray} E[e^{-tz^2}]=\frac{1}{\sqrt{1+2t}}e^{\frac{-E[z]^2t}{1+2t}},\quad t>0, \label{ap3}\end{eqnarray}
to rewrite (\ref{20}) as (\ref{ap2}). To do so, we set
$t=\frac{1}{2}\alpha_x\sigma^2_y=\frac{1}{2}\sigma^2_{\mu}\frac{\alpha_p^2}{\alpha_x}$,
and substitute $z^2=\frac{y^2}{\sigma^2_y}$, with
$y=-\Phi^{-1}(c)\frac{\sqrt{\alpha}}{\alpha_x}+\theta^*\frac{\alpha_p}{\alpha_x}-\frac{\alpha_p}{\alpha_x}\mu$,
into (\ref{20}), to obtain:\footnote{Regarding the equality in
(\ref{ap4}), we note that
$\frac{\alpha_p}{\sqrt{\alpha_x}}\frac{1}{\sqrt{2\pi}}\int_{-\infty}^{\infty}e^{-tz^2}\phi(\mu)d\mu
=\frac{\alpha_p}{\sqrt{\alpha_x}}\frac{1}{\sqrt{2\pi}}\int_{-\infty}^{\infty}e^{-tz^2}\frac{1}{\sqrt{2\pi}\sigma_{\mu}}e^{-\frac{(\mu-E[\mu])^2}{2\sigma^2_{\mu}}}d\mu
=\frac{\alpha_p}{\sqrt{\alpha_x}}\frac{1}{\sqrt{2\pi}}\int_{-\infty}^{\infty}e^{-tz^2}\frac{1}{\sqrt{2\pi}\sigma_{\mu}}e^{-\frac{(z-E[z])^2}{2}}\sigma_{\mu}dz
=\frac{\alpha_p}{\sqrt{\alpha_x}}\frac{1}{\sqrt{2\pi}}\int_{-\infty}^{\infty}e^{-tz^2}\frac{1}{\sqrt{2\pi}}e^{-\frac{(z-E[z])^2}{2}}dz$.
It now follows from the steps in footnote \ref{f1}, equation
(\ref{ap3}) respectively, that the equality in (\ref{ap4}) holds.}
\begin{eqnarray}
\frac{\alpha_p}{\sqrt{\alpha_x}}\frac{1}{\sqrt{2\pi}}\int_{-\infty}^{\infty}e^{-tz^2}\phi(\mu)d\mu
=_{|(\ref{ap3})}\frac{\alpha_p}{\sqrt{\alpha_x}}\frac{1}{\sqrt{2\pi}}
\frac{1}{\sqrt{1+\alpha_x\sigma_y^2}}e^{-\frac{E[y]^2}{\sigma^2_y}\frac{1}{2}\alpha_x\sigma^2_y\frac{1}{1+\alpha_x\sigma_y^2}}\label{ap4}\\
=\frac{\alpha_p}{\sqrt{\alpha_x}}\frac{1}{\sqrt{2\pi}}\frac{1}{\sqrt{1+(\sigma_{\mu}\frac{\alpha_p}{\sqrt{\alpha_x}})^2}}
e^{-\frac{1}{2}\Big(\sqrt{\alpha_x}\Big(
-\Phi^{-1}(c)\frac{\sqrt{\alpha}}{\alpha_x}+\theta^*\frac{\alpha_p}{\alpha_x}-\frac{\alpha_p}{\alpha_x}E[\mu]\Big)\Big)^2\frac{1}{1+
(\sigma_{\mu}\frac{\alpha_p}{\sqrt{\alpha_x}})^2}}.
\label{ap5}\end{eqnarray} Where the final step from (\ref{ap4}) to
(\ref{ap5}) involves cancelling and resubstitution of
$\alpha_x\sigma^2_y=(\sigma_{\mu}\frac{\alpha_p}{\sqrt{\alpha_x}})^2$
and
$E[y]=-\Phi^{-1}(c)\frac{\sqrt{\alpha}}{\alpha_x}+\theta^*\frac{\alpha_p}{\alpha_x}-\frac{\alpha_p}{\alpha_x}E[\mu]$.



\newpage


\end{appendix}

\addcontentsline{toc}{section}{References}
\bibliographystyle{apalike}
\bibliography{References}

\end{document}